\documentclass[aps,pre,onecolumn,superscriptaddress,amsmath,amssymb]{revtex4-2}
\usepackage{graphicx}%
\usepackage{amsthm,amsfonts}%
\usepackage{mathrsfs}%
\usepackage{xcolor}%
\usepackage{booktabs}%
\usepackage{url}

\usepackage{dsfont}
\usepackage{orcidlink}

\newcommand{\matr}[1]{{\boldsymbol{#1}}}
\renewcommand{\vec}[1]{{\boldsymbol{#1}}}
\newcommand{\EE}{\mathcal{E}}
\newcommand{\VV}{\mathcal{V}}

\newcommand{\OO}{\mathcal{F}}
\newcommand{\lagrangian}{\mathcal{L}}
\newcommand{\nn}{N}
\newcommand{\nedge}{{N_\EE}}

\newcommand{\eye}{\mathds{1}}

\newtheorem{cor}{Corollary}
\newtheorem{lem}{Lemma}
\newtheorem{thm}{Theorem}

\begin{document}

\title{Linear Approximations Provide Rigorous Conditions for Phase Locking in Oscillator Networks and Power Grids}

\newcommand{\FZJ}{Institute of Climate and Energy Systems: Energy Systems Engineering (ICE-1), Forschungszentrum J\"ulich, J\"ulich, Germany}
\newcommand{\IPB}{Center of Excellence for Study of Complex Systems, Institute of Physics Belgrade, University of Belgrade, Belgrade, Serbia}
\newcommand{\RWTH}{RWTH Aachen University, Aachen, Germany}
\newcommand{\UoC}{Institute for Theoretical Physics, University of Cologne, K\"oln, Germany}

\author{Iva Ba\v{c}i\'{c}\,\orcidlink{0000-0003-2987-5065}}
\affiliation{\FZJ}
\affiliation{\IPB}

\author{Carsten Hartmann\,\orcidlink{0009-0007-5067-589X}}
\affiliation{\FZJ}
\affiliation{\RWTH}

\author{Philipp C.~B\"ottcher\,\orcidlink{0000-0002-3240-0442}}
\affiliation{\FZJ}

\author{David Gross\,\orcidlink{0000-0001-8888-3028}}
\affiliation{\UoC}

\author{Andrea Benigni\,\orcidlink{0000-0002-2475-7003}}
\affiliation{\FZJ}
\affiliation{\RWTH}

\author{Dirk Witthaut\,\orcidlink{0000-0002-3623-5341}}
\affiliation{\FZJ}
\affiliation{\UoC}

\begin{abstract}
Synchronization is a fundamental collective phenomenon in many natural and engineered networks. The stable operation of oscillator networks, such as electric power grids, requires phase locking with phase differences across edges below a given threshold. A central question is whether a given network supports such a state. In this contribution, we extend a convex optimization formulation of the nonlinear fixed point equations of the Kuramoto model, that also describe lossless power grids. Based on this formulation, we derive rigorous and tight bounds on the error introduced by an arbitrary approximation of the original model, in particular linear approximations. The bounds yield edge-wise regions of trust around the approximation, within which the true nonlinear solution must lie. They further provide sufficient and necessary conditions for phase locking under prescribed phase difference limits. Numerical experiments show that the bounds are tight over a broad range of parameters.
\end{abstract}


\maketitle

\section{Introduction}

Synchronization is a fundamental collective phenomenon in complex networks and plays an important role in applications ranging from coupled oscillators and distributed control to infrastructure systems~\cite{arenas2008synchronization,dorfler2014survey}. A central theoretical question is whether a given network admits a phase-locked state and, more stringently, whether such a state satisfies prescribed bounds on the phase differences between adjacent nodes. These questions depend jointly on the network topology, the coupling strengths, and the heterogeneity of the nodal inputs.

Electric power grids provide a particularly important example.
Stable steady-state operation requires generators and grid-forming devices to remain phase locked, thereby enabling a stationary transfer of real power across the network~\cite{witthaut2022collective}. Loss of phase locking between parts of the grid destroys steady power transmission and can contribute to large-scale system failures~\cite{titz2026cascading}. At the same time, increasing utilization of transmission infrastructure raises the loading of individual corridors and reduces operational margins. It is therefore of practical importance to determine whether a network supports a \emph{phase-cohesive} state in which all phase differences remain below a prescribed threshold.

The existence of phase-cohesive states has been studied extensively in oscillator networks and power systems. Effective sufficient conditions were derived in Refs.~\cite{dorfler2013synchronization,jafarpour2018synchronization}, while the role of cycles, loop flows, and winding numbers in meshed networks was analyzed in Refs.~\cite{manik2017cycle,delabays2016multistability,jafarpour2022flow,cui2026cyclesolvability}. Related approaches exploit connections to nonlinear consensus~\cite{jadbabaie2004stability,dorfler2011critical} or formulate the equilibrium equations through convex or energy-based optimization problems~\cite{DvijothamLowChertkov,hartmann2024synchronized}. 
In our recent work~\cite{hartmann2024synchronized}, we introduced a convex formulation of the nonlinear fixed-point equations in terms of edge flows and derived rigorous bounds on the error of the linear (DC) approximation, in which the sine term is replaced by its argument. This approximation is commonly used by transmission system operators to speed up simulations. Here, we build on this framework and derive \emph{edge-wise a posteriori error bounds} for arbitrary feasible approximations, i.e. bounds obtained from the approximation itself, without knowledge of the exact solution. These bounds define a rigorous region of trust around the DC solution and lead directly to new sufficient and necessary conditions for the existence of phase-cohesive states. The resulting conditions explicitly account for cycle flows and network topology and apply equally to phase-locked states of oscillator networks and to lossless real-power flow in electric grids.

\section{Background and Notation}

We consider networks of coupled phase oscillators described by the Kuramoto model~\cite{kuramoto1975,arenas2008synchronization},
\begin{align}
    \dot \theta_n
    =
    \omega_n
    -
    \sum_{m \in \VV}
    K_{nm}\sin(\theta_n-\theta_m),
    \label{eq:kuramoto}
\end{align}
where $\theta_n$ denotes the phase of node $n$, $\omega_n$ its natural frequency, and $K_{nm}=K_{mn}\geq 0$ the coupling strength between nodes $n$ and $m$. 
We assume that the underlying network is connected and denote its node and edge sets by $\VV$ and $\EE$, respectively. 
A phase-locked state is characterized by all oscillators rotating at a common frequency $\Omega = \frac{1}{|\VV|}    \sum_{m\in\VV} \omega_m$ while their phase differences remain constant. Such a state corresponds to a fixed point in a co-rotating frame given by the algebraic equations
\begin{align}
    {\omega}_n - \Omega
    =
    \sum_{m\in\VV}
    K_{nm}\sin(\theta_n-\theta_m),
    \qquad n\in\VV.
    \label{eq:fixed-point}
\end{align}
The same equations describe the phase-locked states of the second-order Kuramoto model~\cite{tanaka1997first}, where the inertial and damping terms reduce to a constant in the co-rotating frame and are absorbed into the effective injection $p_n$.

Equations~\eqref{eq:fixed-point} arise in a variety of networked dynamical systems, such as nonlinear consensus protocols~\cite{jadbabaie2004stability}, vehicle coordination~\cite{klein2008integration}, and clock synchronization in distributed computer networks~\cite{simeone2008distributed}. They also describe the steady-state real-power balance in lossless AC power grids~\cite{dorfler2014survey,witthaut2022collective},
\begin{align}
    p_n
    =
    \sum_{m\in\VV}
    K_{nm}\sin(\theta_n-\theta_m),
    \qquad n\in\VV.
    \label{eq:realpower-intro}
\end{align}
Here, $\theta_n$ denotes the voltage phase angle and $p_n$ the real power injection at bus $n$, while the effective coupling strength $K_ {nm}$ is determined by the parameters of the branch $(n,m)$ and the voltage magnitude at the terminal buses. Power balance requires $\sum_{n\in\VV}p_n=0$. Stable power grid operation generally requires that fixed points are \emph{phase-cohesive}~\cite{dorfler2013synchronization} in the sense that 
\begin{align}
    |\theta_n-\theta_m| \leq \gamma
    \qquad
    \text{for all } (n,m)\in\EE,
    \label{eq:phase-cohesiveness}
\end{align}
with a prescribed threshold $\gamma<\pi/2$. In the following, we adopt the terminology of power-system analysis because it provides a particularly useful interpretation in terms of power injections and branch flows. All results, however, apply equally to phase-locked states of oscillator networks upon identifying
 $   p_n \equiv \omega_n-\Omega .$

A particularly useful perspective is obtained by shifting the description from nodes to edges. For every oriented edge $e=(n,m)$, we define the flow
\begin{align}
    f_e = K_e\sin(\theta_n-\theta_m).
    \label{eq:def-flow-background}
\end{align}
The nonlinear fixed-point problem can then be separated into two conceptually distinct requirements. First, the edge flows must
be conserved  at every node,
\begin{align}
    \vec p = \matr E \vec f,
    \label{eq:power-balance}
\end{align}
where $\matr E$ is the signed node-edge incidence matrix~\cite{newman2018networks}, with components $E_{n,e} =+1$ if node $n$ is the from end of branch $e$, $E_{n,e} =-1$ if it is the to end of branch $e$, and $0$ otherwise. Second, the flows must be compatible with a set of nodal phases through Eq.~\eqref{eq:def-flow-background}. This edge-based formulation is particularly advantageous for theoretical analysis: the linear conservation law determines an affine space of admissible flows, while all effects of the nonlinear coupling are encoded in the additional compatibility condition. The distinction between these two requirements will be central to the analysis below.

A natural approximation to the nonlinear equilibrium equations is obtained by linearizing the sine for small phase differences. In power-system analysis, the resulting model is known as the DC power-flow approximation~\cite{wood2013power}. It is heavily used in time critical applications, such as in congestion management and redispatch, or in N-1 contingency screening, where thousands of scenarios must be evaluated in very short timescales~\cite{purchala2005usefulness,westerbeck2025accelerated}. 

The branch flows thus read
\begin{align}
    &\mbox{real power flow:} &
    \vec f &= \matr K \sin\left( \matr E^\top \vec \theta \right), 
    \label{eq:RP_flow}
    \\
    &\mbox{DC approximation :} &
    \vec f^{\rm DC} &= \matr K  \matr E^\top \vec \theta^{\rm DC} \label{eq:DC_flow},
\end{align}
while linear balance condition $\vec p = \matr E \vec f$ must be satisfied in all cases.

\section{Algebraic graph theory of network flows}
\label{sec:graph-theory}

We briefly summarize the graph-theoretic tools required below; see, e.g.,~\cite{newman2018networks,jafarpour2018synchronization,hartmann2024synchronized} for further details. 
The weighted graph Laplacian associated with the coupling strengths is
   $ \matr L = \matr E \matr K \matr E^\top $, 
where $\matr K=\operatorname{diag}(K_1,\ldots,K_{\nedge})$, $\nedge = |\EE|$ is the number of edges, and $N=|\VV|$ is the number of nodes.
For a connected network, $\ker(\matr L)$ is spanned by the constant vector. Hence, for balanced injections $\sum_n p_n=0$, the DC approximation has the solution
 $   \vec\theta^{\rm DC} = \matr L^+\vec p$,
   $ \vec f^{\rm DC} = \matr K\matr E^\top\matr L^+\vec p $,
where ${}^+$ denotes the Moore--Penrose pseudoinverse.

The balance equation $\matr E\vec f=\vec p$ does not determine the edge flows uniquely on a meshed network. If $\vec f$ and $\vec f'$ satisfy the same balance equation, their difference lies in the cycle space
$
    \ker(\matr E).
$
Choosing a basis of $\nedge-\nn+1$ independent cycles, we collect the corresponding oriented edge vectors in the cycle-edge incidence matrix
$\matr C\in\mathbb{R}^{\nedge\times(\nedge-\nn+1)}$, such that
$
    \operatorname{range}(\matr C)=\ker(\matr E).
$
Thus, every flow satisfying the balance equation can be written as
   $ \vec f=\vec f^{\rm DC}+\matr C\vec\ell $, 
where $\vec\ell$ contains the loop-flow amplitudes. The projector onto the cycle space~\cite{jafarpour2018synchronization,hartmann2024synchronized},
$    \matr\Pi_l
    =
    \eye-\matr K\matr E^\top\matr L^+\matr E $
satisfies
    $\matr\Pi_l^2 = \matr\Pi_l$,
    $\matr E\matr\Pi_l =0$, and 
    $\matr K\matr\Pi_l^\top = \matr\Pi_l\matr K$, 
and is orthogonal with respect to the weighted inner product
    $\langle\vec\xi,\vec\zeta\rangle_K
    :=
    \vec\xi^\top\matr K^{-1}\vec\zeta$, with the norm 
    $\|\vec\xi\|_K^2
    :=
    \langle\vec\xi,\vec\xi\rangle_K $. 
These properties will allow us to isolate precisely the part of an approximation error that lies in the cycle space.

Finally, not every flow satisfying $\matr E\vec f=\vec p$ is compatible with a set of nodal phases. For flows with $|f_e|<K_e$, phase compatibility requires~\cite{manik2017cycle,jafarpour2022flow}
\begin{align}
    \matr C^\top
    \arcsin\!\left(\matr K^{-1}\vec f\right)
    =2\pi \, \vec z,
    \qquad
    \vec z\in\mathbb Z^{\nedge-\nn+1}.
    \label{eq:winding-condition}
\end{align}
The integer vector $\vec z$ is the winding vector. Physically, $z_\alpha$ counts the net phase winding around cycle $\alpha$: traversing a closed loop, the phase differences must add up to an integer multiple of $2\pi$.

Different winding vectors can therefore correspond to distinct phase-locked states with the same nodal injections. Figure~\ref{fig:Multistability} illustrates this phenomenon for a test network with 14 nodes~\cite{matpower2011,pandapower2018}, where two stable solutions coexist and differ by a loop flow on the central cycle. In this article, we focus on the physically most relevant branch $\vec z=\vec 0$.

\begin{figure}[tb]
\centering
\includegraphics[width=.6\columnwidth]{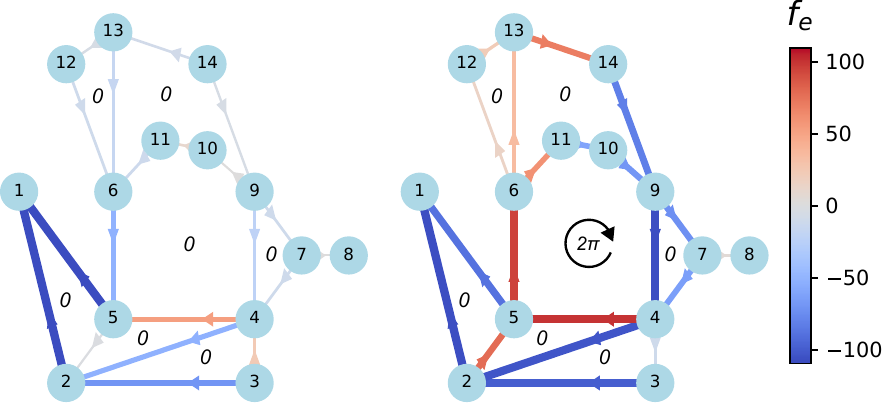}
\caption{Two coexisting phase-locked states in a modified IEEE \texttt{case14} network. Edge widths and colors indicate $|f_e|$, arrows show the flow direction, and the numbers inside the faces indicate the corresponding cycle sums
$\matr C^\top\arcsin(\matr K^{-1}\vec f)$. The solutions differ in the winding number of the central cycle. All edges are assumed to be lossless and have the same $K_e = 110.14$ MW.
}
\label{fig:Multistability}
\end{figure}

\section{Convex formulation of nonlinear network equilibria}
\label{sec:optimization}

We now return to the nonlinear fixed-point equations \eqref{eq:realpower-intro} and ask when they admit a phase-cohesive solution satisfying $|\theta_n-\theta_m|\leq\gamma<\pi/2$ on every edge. Our approach builds on the convex flow formulation introduced in Ref.~\cite{hartmann2024synchronized}. There, the nonlinear fixed-point problem was reformulated as a convex optimization problem and used to derive global bounds on the error of the DC approximation. Here, we extend this formulation to obtain sharper edge-wise error bounds and, ultimately, explicit existence and nonexistence certificates for phase-cohesive states.

For this purpose, consider the optimization problem
\begin{align}
    \label{opt:realpower1}
    & \min_{\vec f} \OO(\vec f)
    \\
    & \mathrm{s.t.}\quad
    \vec p = \matr E\vec f ,
    \label{opt:affine}
\end{align}
with the separable objective
\begin{align}
    \OO(\vec f) &= \sum_{e\in\EE}\OO_e(f_e),
    \label{eq:objective-fun}\\
    \OO_e(f_e)
    &=
    \tilde f_e
    \arcsin\!\left(\frac{\tilde f_e}{K_e}\right)
    +\sqrt{K_e^2-\tilde f_e^2}
    +\pi n_e(f_e-n_eK_e),
    \label{opt:objfunc}
\end{align}
where
   $ \tilde f_e = f_e-2n_eK_e$, and
   $ n_e = \left\lfloor\frac{f_e/K_e+1}{2}\right\rfloor .$ 
The functions $\OO_e$ are continuously differentiable and strictly
convex, and hence the optimization problem has a unique minimizer.
The periodic continuation in Eq.~\eqref{opt:objfunc} extends the
objective introduced in Ref.~\cite{hartmann2024synchronized} beyond the
principal interval $|f_e|<K_e$.

The relation between the minimizer and the nonlinear fixed point is
summarized by the following result.

\begin{thm}
\label{thm:realpower1}
Let $\gamma<\pi/2$. The following statements are equivalent:
\begin{enumerate}
    \item The vector $\vec\theta$ solves Eq.~\eqref{eq:realpower-intro} with winding vector
    $\vec z=\vec 0$ and satisfies
    $|\theta_n-\theta_m|\leq\gamma$ for all $(n,m)\in\EE$.
    \item The associated flow
      $  \vec f=\matr K\sin(\matr E^\top\vec\theta)$ 
    is the unique minimizer of Eq.~\eqref{opt:realpower1} and satisfies
      $  |f_e| \leq K_e\sin\gamma$
        for all $e \in \EE$.
\end{enumerate}
\end{thm}

\begin{proof}
For $|f_e|<K_e$, Eq.~\eqref{opt:objfunc} reduces to the objective introduced in Ref.~\cite{hartmann2024synchronized}, whose stationarity condition is $ \arcsin(f_e/K_e)     =(\matr E^\top\vec\lambda)_e$ where $\vec\lambda$ is a vector of Lagrange multipliers. Hence,     $\vec f     =     \matr K\sin(\matr E^\top\vec\lambda)$,
and substitution into the constraint $\vec p=\matr E\vec f$ reproduces Eq.~\eqref{eq:realpower-intro} upon identifying $\vec\lambda$ with $\vec\theta$. The minimizer belongs to the branch $\vec z=\vec 0$; see Ref.~\cite{hartmann2024synchronized} for the detailed argument.
Since $\gamma<\pi/2$, the sine is strictly monotone on
$[-\gamma,\gamma]$, and therefore
\begin{align}
    |\theta_n-\theta_m|\leq\gamma
    \quad\Longleftrightarrow\quad
    |f_e|/K_e\leq\sin\gamma .
\end{align}
\end{proof}

We restrict attention to the branch $\vec z=\vec 0$ throughout the remainder of the paper. Extensions 
to nonzero winding vectors are discussed in Ref.~\cite{hartmann2024synchronized}.

\section{A posteriori error bounds}

We now derive rigorous a posteriori bounds on the error of an arbitrary approximate flow $\vec f^\circ$ satisfying the balance equation $\matr E\vec f^\circ=\vec p$. The bounds are derived from the approximation itself, without solving the nonlinear problem. 
Our derivation extends the error analysis of Ref.~\cite{hartmann2024synchronized} in two respects: it applies to a general feasible approximation and, most importantly, yields edge-wise error bounds that we will use to certify phase-cohesiveness.

We denote the gradient and Hessian of the objective by
\begin{align}
    g_e(\vec f)
    &= \frac{\partial\OO}{\partial f_e}
     = \arcsin(f_e/K_e-2n_e)+\pi n_e ,
    \\
    H_{e,e'}(\vec f)
    &=\delta_{e,e'}
      \frac{1}{\sqrt{K_e^2-(f_e-2n_eK_e)^2}} .
\end{align}
The Hessian is diagonal and satisfies
  $  \matr H(\vec f)\succeq \matr K^{-1}.$

\begin{lem}[A posteriori error inequality]
\label{lem:error-inequality}
Let $\vec f^*$ be the unique minimizer of
Eq.~\eqref{opt:realpower1}, and let $\vec f^\circ$ be any feasible
approximation satisfying $\matr E\vec f^\circ=\vec p$. Then the approximation error $\vec\xi:=\vec f^\circ-\vec f^*$ satisfies
\begin{align}
    \vec\xi^\top\matr K^{-1}\vec\xi
    \le
    \vec g(\vec f^\circ)^\top\vec\xi .
    \label{eq:inequality-from-Taylor}
\end{align}
\end{lem}

\begin{proof}

Both $\vec f^\circ$ and $\vec f^*$ satisfy the constraint, hence $\matr E \vec\xi=\vec0$ and we can write $\vec\xi=\matr C \vec \ell$. This parametrizes the feasible set, so $\vec \ell=\vec0$ is an unconstrained minimizer of $\vec \ell \mapsto \OO(\vec f^*+\matr C \vec \ell)$ and the first-order derivative vanishes at this point. Taylor's theorem in Lagrange form then gives
$\OO(\vec f^\circ)=\OO(\vec f^*)+\frac{1}{2}\vec\ell^\top\tilde{\matr{H}}\vec\ell$,
where $\tilde{\matr{H}}=\matr{C}^\top\matr{H}\matr{C}$ is evaluated at an intermediate point. Using $\vec\ell^\top \tilde{\matr{H}}\vec\ell=\vec \xi^\top\matr{H}\vec \xi$ and $\matr{H} \succeq \matr{K}^{-1}$, we obtain $\OO(\vec f^\circ) \ge \OO(\vec f^*) + \frac{1}{2} \vec\xi^\top \matr{K}^{-1} \vec \xi$.
Expanding the same way around $\vec f^\circ$ yields $\OO(\vec f^*)\ge\OO(\vec f^\circ)-\vec g(\vec f^\circ)^\top\vec\xi + \frac{1}{2}\vec\xi^\top\matr{K}^{-1}\vec\xi$. Adding both inequalities gives Eq.~\eqref{eq:inequality-from-Taylor}.
\end{proof}

Lemma~\ref{lem:error-inequality} can then be used to derive a posterior error bounds, either for the norm of the error~\cite{hartmann2024synchronized,cui2026cyclesolvability}, or for individual edge flows.

\begin{cor}[Global error bound]
Let $\vec f^*$ be the unique minimizer of
Eq.~\eqref{opt:realpower1}, and let $\vec f^\circ$ be any feasible
approximation satisfying $\matr E\vec f^\circ=\vec p$. Then the approximation error $\vec\xi=\vec f^\circ-\vec f^*$ satisfies
\begin{align}
    \| \vec \xi \|_K^2  &\le  
    \| \matr \Pi_l\matr K \vec g(\vec f^\circ) \|_K^2 
\end{align}
\end{cor}
\begin{proof}
 Since $\vec\xi=\matr\Pi_l\vec\xi$ and $\matr K\matr\Pi_l^\top=\matr\Pi_l\matr K$, we have
$
    \vec g^\top\vec\xi
    =
    \left\langle
       \vec\xi,\matr\Pi_l\matr K\vec g
    \right\rangle_K .
$ 
Inserting in Eq.~\eqref{eq:inequality-from-Taylor} and applying the Cauchy--Schwarz inequality then proves the result.
\end{proof}

\begin{thm}[Edge-wise error bounds]
\label{thm:xi-a}
Let $\vec f^*$ be the unique solution of the optimization problem \eqref{opt:realpower1}, and let $\vec f^\circ$ be an arbitrary feasible approximation satisfying $\matr E\vec f^\circ=\vec p$ and $\vec \xi = \vec f^\circ-\vec f^*$ the approximation error. Then, for every edge $e\in\EE$,
\begin{align}
    -\chi_e
    \le
    \xi_e-\kappa_e
    \le
    \chi_e ,
    \label{eq:individual-error}
\end{align}
where, for every edge that is not a bridge,
\begin{align}
    \kappa_e
    =
    \frac{1}{2}
    \vec w_e^\top
    \matr \Pi_l\matr K
    \vec g(\vec f^\circ), \qquad
    \chi_e
    =
    \frac{1}{2}
    \sqrt{
        K_e\Pi_{e,e}\,
        \vec g(\vec f^\circ)^\top
        \matr \Pi_l\matr K
        \vec g(\vec f^\circ)
    } .
    \label{eq:deltaxi}
\end{align}
Here, $\vec w_e$ denotes the $e$-th standard basis vector,
$(\vec w_e)_{e'}=\delta_{e,e'}$, and $\Pi_{e,e}$ denotes the $e$-th
diagonal element of $\matr\Pi_l$. For every bridge,
$\kappa_e=\chi_e=0$.
\end{thm}

\begin{proof}
If $e$ is a bridge, every cycle flow vanishes on $e$. Since
$\vec\xi\in\ker(\matr E)$, it follows that $\xi_e=0$, and hence
$\kappa_e=\chi_e=0$. For all other edges, Lemma~\ref{lem:error-inequality} gives
\begin{align}
    \vec\xi^\top\matr K^{-1}\vec\xi
    -
    \vec g^\top\vec\xi
    \le 0,
    \label{eq:inequality2}
\end{align}
where we use the shorthand
$\vec g=\vec g(\vec f^\circ)$.
For a prescribed value of $\xi_e$, consider the optimization problem
\begin{align}
    R_e(\xi_e)
    ={}&
    \min_{\vec\zeta}
    \left[
        \vec\zeta^\top\matr K^{-1}\vec\zeta
        -
        \vec g^\top\vec\zeta
    \right]
    \nonumber\\
    &\mathrm{s.t.}\qquad
    \matr E\vec\zeta=\vec0,
    \qquad
    \zeta_e=\xi_e .
    \label{eq:optimization-R}
\end{align}
A value $\xi_e$ is compatible with Eq.~\eqref{eq:inequality2} only if
$R_e(\xi_e)\le0$.

Introducing Lagrange multipliers $\vec\lambda$ and $\mu$, the
Lagrangian reads
\begin{align}
    \lagrangian(\vec\zeta,\vec\lambda,\mu)
    ={}&
    \vec\zeta^\top\matr K^{-1}\vec\zeta
    -
    \vec g^\top\vec\zeta
    +
    \vec\lambda^\top\matr E\vec\zeta
    +
    \mu(\zeta_e-\xi_e).
\end{align}
Stationarity with respect to $\vec\zeta$ yields
  $  \vec\zeta^*
    =
    \frac{1}{2}\matr K
    \left(
        \vec g
        -
        \matr E^\top\vec\lambda
        -
        \mu\vec w_e
    \right).$ 
Imposing the constraint
$\matr E\vec\zeta^*=\vec0$ gives
 $   \vec\lambda
    =
    \matr L^+
    \matr E\matr K
    \left(
        \vec g-\mu\vec w_e
    \right).$ 
Substituting $\vec \lambda$ into
$\vec\zeta^*$, and using
$
    \matr\Pi_l
    =
    \eye
    -
    \matr K\matr E^\top
    \matr L^+\matr E,
$
we obtain
$
    \vec\zeta^*
    =
    \frac{1}{2}
    \matr\Pi_l\matr K
    \left(
        \vec g-\mu\vec w_e
    \right).
    \label{eq:zetamin2}
    $ 
The constraint $\zeta_e^*=\xi_e$ then yields
 $   \mu
    =
    \frac{
        \vec w_e^\top
        \matr\Pi_l\matr K\vec g
        -
        2\xi_e
    }{
        K_e\Pi_{e,e}
    } .$ 

Using the projector identities
$\matr\Pi_l^2=\matr\Pi_l$ and
$\matr K\matr\Pi_l^\top=\matr\Pi_l\matr K$,
the condition $R_e(\xi_e)\le0$ reduces to
  $  -
    \vec g^\top
    \matr\Pi_l\matr K
    \vec g
    +
    \mu^2 K_e\Pi_{e,e}
    \le0 .$ 
Substituting $\mu$ gives
 $   \left(
        \xi_e
        -
        \frac{1}{2}
        \vec w_e^\top
        \matr\Pi_l\matr K
        \vec g
    \right)^2
    \le
    \frac{1}{4}
    K_e\Pi_{e,e}\,
    \vec g^\top
    \matr\Pi_l\matr K
    \vec g .$ 
With the definitions~\eqref{eq:deltaxi}, this is equivalent to
Eq.~\eqref{eq:individual-error}.
\end{proof}

Theorem~\ref{thm:xi-a} provides an edge-wise region of trust around
the approximate flow. Since
$\vec f^*=\vec f^\circ-\vec\xi$, Eq.~\eqref{eq:individual-error}
implies
$    f_e^*
    \in
    \left[
        f_e^\circ-\kappa_e-\chi_e,\,
        f_e^\circ-\kappa_e+\chi_e
    \right].$ 
Importantly, the center $\kappa_e$ and radius $\chi_e$ of this interval depend only on the approximation $\vec f^\circ$, the coupling strengths, and the network topology. The exact nonlinear solution $\vec f^*$ is \emph{not} required.

\section{Sufficient and necessary conditions for phase cohesiveness from the DC approximation}
\label{sec:existence}

We now focus on the DC approximation $\vec f^{\rm DC}$. 
Theorem~\ref{thm:xi-a} guarantees that the nonlinear flow $\vec f^*$ lies in the region of trust
$    f_e^*
    \in
    \left[
        f_e^{\rm DC}-\kappa_e-\chi_e,\,
        f_e^{\rm DC}-\kappa_e+\chi_e
    \right].$
A phase-cohesive solution with maximum phase difference $\gamma<\pi/2$
must satisfy
   $ |f_e^*| \le K_e\sin\gamma$, for every $e \in \EE$. 
Comparing these two intervals directly yields sufficient and necessary conditions.

\begin{cor}[Sufficient condition]
\label{cor:psi-lin-psi-rp}
Let $\vec f^{\rm DC}$ denote the DC approximation and let
$\gamma<\pi/2$. If
   $ |f_e^{\rm DC}-\kappa_e|+\chi_e
    \le K_e\sin\gamma$ for every $e \in \EE$,
then the nonlinear fixed-point equations~\eqref{eq:realpower-intro}
admit a solution with $\vec z=\vec0$ satisfying
  $  |\theta_n-\theta_m|\le\gamma$ for every $(n, m) \in \EE$. 
\end{cor}

Corollary~\ref{cor:psi-lin-psi-rp} is closely related to the synchronization condition of Ref.~\cite{dorfler2013synchronization}. On trees, $\matr\Pi_l=\vec0$ and hence $\kappa_e=\chi_e=0$, so the conditions coincide. On meshed networks, the present condition explicitly accounts for loop-flow corrections.  See~\cite{Bacic2025phasecohesive} for a detailed discussion.

\begin{cor}[Necessary condition]
\label{cor:psi-lin-necessary}
Let $\vec f^{\rm DC}$ denote the DC approximation and let
$\gamma<\pi/2$. If a solution with $\vec z=\vec0$ satisfying
$|\theta_n-\theta_m|\le\gamma$ on all edges exists, then
 $   |f_e^{\rm DC}-\kappa_e|-\chi_e
    \le K_e\sin\gamma$ for every $e \in \EE$.
\end{cor}

Equivalently, if this condition is violated on any edge, then no phase-cohesive solution with winding vector $\vec z=\vec0$ exists.

\section{Numerical validation}
\label{sec:numerical_val_1}

We validate the bounds on modified \textsc{Matpower} \texttt{case30} and \texttt{case118} test systems~\cite{matpower2011,pandapower2018}. Ohmic losses are neglected, and the effective coupling strengths are fixed at their reference values. To vary the overall network loading, all power injections are multiplied by a common factor $p_f$.

\begin{figure*}[tb]
\includegraphics[width=0.7\linewidth]{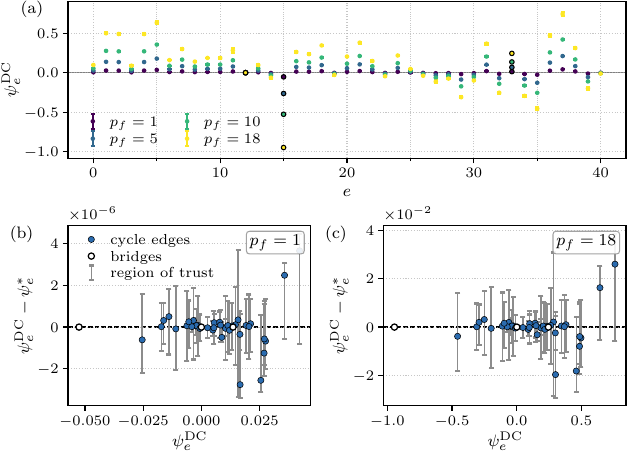}
\caption{(a) Region of trust $\psi^{\mathrm{DC}}_e - \kappa_e/K_e \pm \chi_e/K_e$ for every edge $e$ and four loading factors $p_f$ for the modified \texttt{case30}. The intervals are narrower than the symbol size. Bridges are marked with black outline. 
(b,c) Region of trust for $p_f=1$ and $p_f=18$. 
Symbols show the error of the DC approximation relative to the numerically exact solution of Eq.~\eqref{eq:fixed-point}, $\psi^{\mathrm{DC}}_e-\psi^*_e$, for cycle edges (filled) and bridges (white). 
Bars show the region of trust $[\kappa_e - \chi_e, \kappa_e + \chi_e]/K_e $. Every error lies within the region of trust.}
\label{fig:trustregion}
\end{figure*}

Figure~\ref{fig:trustregion} illustrates the edge-wise region of trust from Theorem~\ref{thm:xi-a} for \texttt{case30}, in terms of $\psi_e = f_e/K_e = \sin(\theta_n-\theta_m)$. For every edge, the exact nonlinear flow $f_e^*$ lies within
\begin{align}
    f_e^{\rm DC}-\kappa_e-\chi_e
    \le f_e^*
    \le
    f_e^{\rm DC}-\kappa_e+\chi_e .
\end{align}
Although the error increases with $p_f$, the intervals remain narrow even for strongly increased loading, showing that the DC approximation provides an accurate reference point for the subsequent analysis. 
To quantify the resulting bound on phase cohesiveness, we define
\begin{align}
    \gamma
    &=
    \max_{(n,m)\in\EE}
    |\theta_n^*-\theta_m^*|, \qquad \qquad
    \Psi
    :=
    \max_{e\in\EE}
    \frac{|f_e^{\rm DC}-\kappa_e|+\chi_e}{K_e}.
    \label{eq:Psi-def}
\end{align}
Theorem~\ref{thm:xi-a} directly implies  $\sin\gamma \le \Psi $ and thus $\Psi/\sin\gamma \ge 1$.

Figure~\ref{fig:phasecohesive} compares the exact value
$\sin\gamma$ with the bound $\Psi$. For the standard injection pattern,
the agreement is very close over a broad range of loading levels.
For \texttt{case30}, the bound is exact because the maximally loaded
edge is a bridge, for which $\kappa_e=\chi_e=0$. For
\texttt{case118}, the bound remains tight as long as the system stays
well inside the phase-cohesive regime.

\begin{figure*}[tb]
\centering
\includegraphics[width=0.75\linewidth]{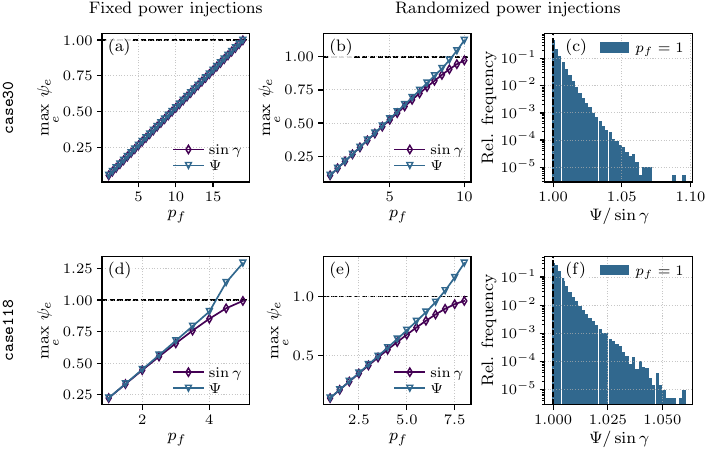}
\caption{Exact phase cohesiveness $\sin\gamma$ and its upper bound
$\Psi$ for \texttt{case30} (a--c) and \texttt{case118} (d--f). In panel (a), the maximum flow is over a bridge, hence $\sin \gamma$ and $\Psi$ coincide. 
(a,d) Standard power injections with increasing loading factor $p_f$. Dashed line marks $\gamma = \pi/2$. 
(b,e) Average over 200 randomized balanced injection patterns.
(c,f) Distribution of $\Psi/\sin\gamma$ for
$2\times10^5$ random balanced injections at $p_f=1$.}
\label{fig:phasecohesive}
\end{figure*}

We further test the bounds for randomized balanced injection vectors
$\vec p$, generated as in Ref.~\cite{hartmann2024synchronized}.
For the loading sweeps in Fig.~\ref{fig:phasecohesive}(b,e), we average
over 200 realizations that admit solutions throughout the investigated
range of $p_f$. Panels (c,f) show the ratio
$\Psi/\sin\gamma$ for $2\times10^5$ random realizations at $p_f=1$.
In both test systems, the distribution is strongly concentrated near
unity, demonstrating that the rigorous bound is close to the
actual maximum phase difference.

\section{Conclusion}

We have derived rigorous a posteriori bounds on approximation errors for nonlinear phase-locked states in complex networks, by building on the convex flow formulation of Ref.~\cite{hartmann2024synchronized}. The resulting edge-wise regions of trust quantify the deviation of an arbitrary feasible approximation from the exact nonlinear solution and explicitly account for the cycle structure of the network.

Applied to the DC approximation, these bounds yield simple sufficient and necessary conditions for the existence of phase-cohesive states with prescribed maximum phase differences. Numerical experiments on benchmark power-grid networks show that the bounds are typically tight over a broad range of operating conditions. The framework is formulated in terms of generic oscillator-network equilibria and therefore applies beyond power systems to other synchronization problems on complex networks.

\end{document}